\documentclass[12pt]{article}

\usepackage{jfe_preambule}
\usepackage{booktabs}

\begin{document}

\setlist{noitemsep}  

\title{Connecting Sharpe ratio and Student t-statistic, and beyond}

\author{Eric Benhamou 
\thanks{A.I. SQUARE CONNECT, 35 Boulevard d'Inkermann 92200 Neuilly sur Seine, France}  
\textsuperscript{,} 
\thanks{LAMSADE, Université Paris Dauphine, Place du Maréchal de Lattre de Tassigny,75016 Paris, France} 
\textsuperscript{,} 
\thanks{E-mail: eric.benhamou@aisquareconnect.com, eric.benhamou@dauphine.eu}
}

\date{}              

\singlespacing

\maketitle

\vspace{-.2in}
\begin{abstract}
\noindent Sharpe ratio is widely used in asset management to compare and benchmark funds and asset managers. It computes the ratio of the excess return over the strategy standard deviation. However, the elements to compute the Sharpe ratio, namely, the expected returns and the volatilities are unknown numbers and need to be estimated statistically. This means that the Sharpe ratio used by funds is subject to be error prone because of statistical estimation error. \cite{Lo_2002}, \cite{Mertens_2002} derive explicit expressions for the statistical distribution of the Sharpe ratio using standard asymptotic theory under several sets of assumptions (independent normally distributed - and identically distributed returns). In this paper, we provide the exact distribution of the Sharpe ratio for independent normally distributed return. In this case, the Sharpe ratio statistic is up to a rescaling factor a non centered Student distribution whose characteristics have been widely studied by statisticians. The asymptotic behavior of our distribution provides the result of \cite{Lo_2002}. We also illustrate the fact that the empirical Sharpe ratio is asymptotically optimal in the sense that it achieves the Cramer Rao bound. We then study the empirical SR under AR(1) assumptions and investigate the effect of compounding period on the Sharpe (computing the annual Sharpe with monthly data for instance). We finally provide general formula in this case of heteroscedasticity and autocorrelation.
\end{abstract}

\medskip

\noindent \textit{JEL classification}: C12, G11.

\medskip
\noindent \textit{Keywords}: Sharpe ratio, Student distribution, compounding effect on Sharpe, AR(1), Cramer Rao bound

\clearpage

\onehalfspacing
\setcounter{footnote}{0}
\renewcommand{\thefootnote}{\arabic{footnote}}
\setcounter{page}{1}


\section{Introduction}
When facing choices to invest in various funds (whether mutual or hedge funds), it is quite common to compare their Sharpe ratio in order to rank funds. This indicator aims at measuring performance for a given risk. This eponymous ratio established by \cite{Sharpe_1966} is a simple number easy to understand. It computes the ratio of the excess return over the strategy standard deviation. However, the elements to compute the Sharpe ratio, namely, the expected returns and the volatilities are unknown numbers and need to be estimated statistically. This means that the Sharpe ratio used by funds is subject to be error prone because of statistical estimation error. In a seminal paper, \cite{Lo_2002} derive explicit expressions for the statistical distribution of the Sharpe ratio using standard asymptotic theory under several sets of assumptions (independent normally distributed - and identically distributed returns). This is interesting as it provides intuition of potential bias and correction to apply to get an unbiased estimator. However, the results are provided as asymptotic results. It could be interested to derive or extend result to the non asymptotic distribution. This is precisely the contribution of this paper. First, we extend provide the exact distribution of the Sharpe ratio for independent normally distributed return. We show that under these conditions, the Sharpe ratio statistic is a non centered Student distribution whose characteristics have been widely studied by statisticians, up to a rescaling factor. Results of \cite{Lo_2002} are easily derived as the limit of our results when the sample size tends to be large. We also study the asymptotic efficiency of the Sharpe ratio statistics and provide exact distribution under AR(1) normal process conditions. We examine finally the impact of compounding effect for computing the Sharpe ratio. The standard square root rule is questionable as soon as there is autocorrelation or homoscedasticity.

\section{Primer on Student distribution}\label{Student}
\subsection{Historical anecdote: why Student?}
The Student t-distribution has been widely studied in statistics. Originally derived as a posterior distribution in 1876 by \cite{Helmert_1876} and \cite{Luroth_1876} as well as in a more general form as Pearson Type IV distribution in \cite{Pearson_1985}, the Student distribution was really popularized by William Sealy Gosset in \cite{Student_1908}. There is various interpretations why this distribution has been published under the pseudonym 'Student'. 

Gosset, an Oxford graduate, worked at the Guinness Brewery in Dublin, Ireland. He was interested in testing small samples – for instance, he wanted to test the chemical properties of barley where sample sizes might be as few as 3. The main version of the origin of the pseudonym is that Gosset's employer, Guiness, preferred staff to use pen names to keep secret their inventions. So when publishing scientific papers, instead of their real name, scientist used different names to hide their identity. Gosset chose the name "Student". Posterity kept this name for the distribution. Another version is that Guinness was reluctant to make public to their competitors their usage of the t-test to determine the quality of raw materials. 
All in all,  Gosset's most noteworthy achievement is now called Student's, rather than Gosset's, t-distribution. 

\subsection{Assumptions}
If $X, ..., X_n$ are independent and identically distributed as a normal distribution with mean $\mu$ and variance $\sigma^2$, 
then the empirical average 
\[
\bar{X} = \frac {1} {n} \sum_{i=1}^n X_i
\]
follows also a normal distribution. The empirical (Bessel-corrected) variance 
\[
\hat{\sigma}^2 = \frac 1 {n-1} \sum_{i=1}^n (X_i - \bar X)^2 
\]
follows (up to the renormalizing term $n-1$) a Chi Square distribution with $n-1$ degree of freedom. The t- statistic defined as
\begin{equation}\label{tstatistic}
\frac{ \bar X -  \mu_0} {\hat{\sigma} / \sqrt{n}} = \sqrt{n}  \,\,\, \frac{ \bar X - \hat \mu_0} { \hat{\sigma}  }
\end{equation}

has a Student's t-distribution with $n-1$ degrees of freedom. 
If the variables $(X_i)_{i=1..n}$ have a mean $\mu$ different from $\mu_0$, the distribution is referred to as a non-central t-distribution with non centrality parameter given by 
\begin{equation}
\eta = \sqrt n \quad \frac{\mu - \mu_0}{\sigma}
\end{equation}
It is simply the expectation of the estimator. The centered and non centered Student distribution are very well known.  
Extension to weaker condition for the t-statistics has been widely studied. \cite{Mauldon_1956} raised the question 
for which pdfs the t-statistic as defined by \ref{tstatistic} is t-distributed with $n - 1$ degrees of freedom. This characterization problem can be generalized to the one of finding all the pdfs for which a certain statistic possesses the property which is a characteristic for these pdfs. \cite{Kagan_1973}, \cite{Bondesson_1974} and \cite{Bondesson_1983} to cite a few tackled Mauldon’s
problem. \cite{Bondesson_1983} proved the necessary and sufficient condition for a t-statistic to have Student’s t-distribution with $n - 1$ degrees
of freedom for all sample sizes is the normality of the underlying distribution. It is not necessary that $X_1,...,X_n$  is an independent sample. Indeed consider
$X_1,...,X_n$  as a random vector $X_n = (X_1,...,X_n)^T$ each component of which having the
same marginal distribution function, $F(·)$. \cite{Efron_1969} has pointed out that the weaker condition
of symmetry can replace the normality assumption. Later, \cite{Fang_2001}  showed that if the vector $X_n$ has a spherical distribution, then the
t-statistic has a t-distribution. A possible extension of
Mauldon’s problem is to find all $F(·)$ for which the Student’s t-statistic has the t-distribution
with $n - 1$ degrees of freedom. Another extension is to determine the distribution of the t-statistic under weaker conditions.
\\ \\

\subsection{A few properties}
Its cumulative distribution function can be expressed in closed form (see for instance \cite{Lenth_1989}) as follows:

\begin{equation}
F_{n-1 ,\eta }(x)= 
\begin{cases}{
\tilde{F}}_{n-1 ,\eta }(x),&{\mbox{if }}x\geq 0;\\
1-{\tilde  {F}}_{{n-1 ,-\eta }}(x),&{\mbox{if }}x<0,
\end{cases}
\end{equation}

\noindent where
${\tilde  {F}}_{n-1 ,\eta}(x)=\Phi (-\eta ) + \frac  {1}{2} \sum _{j=0}^{\infty }\left[p_{j}I_{y}\left(j+{\frac  {1}{2}},{\frac  {n-1 }{2}}\right)+q_{j}I_{y}\left(j+1,{\frac  {n-1 }{2}}\right)\right]$, \\
$I_{y}\,\!(a,b)$ is the regularized incomplete beta function,\\
$y=\frac  {x^{2}}{x^{2}+n-1 }$,\\
$p_{j}=\frac{1}{j!} \exp \left\{-{\frac  {\eta ^{2}}{2}}\right\}\left({\frac  {\eta ^{2}}{2}}\right)^{j}$,\\
$q_{j}={\frac{\eta }{{\sqrt  {2}}\Gamma (j+3/2)}}\exp \left\{-{\frac  {\eta ^{2}}{2}}\right\}\left({\frac  {\eta ^{2}}{2}}\right)^{j}$,\\
and $\Phi$ is the cumulative distribution function of the standard normal distribution. \\\\
Its  probability density function can be expressed in several forms. The most common form (as implemented in R) is the following:
\begin{equation}
f(x)={
\begin{cases}{\frac {n-1 }{x}}\left\{F_{{n+1,\eta }}\left(x{\sqrt {1+{\frac {2}{n-1 }}}}\right)-F_{{n-1 ,\eta }}(x)\right\},&{\mbox{if }}x\neq 0;\\{\frac {\Gamma ({\frac {n}{2}})}{{\sqrt {\pi n-1 }}\Gamma ({\frac {n-1 }{2}})}}\exp \left(-{\frac {\eta ^{2}}{2}}\right),&{\mbox{if }}x=0.
\end{cases}}
\end{equation}

\noindent In general, the kth raw moment of the non-central t-distribution is

\begin{equation}
{\mbox{E}}\left[T^{k}\right]={
\begin{cases}\left({\frac  {n-1 }{2}}\right)^{{{\frac  {k}{2}}}}{\frac  {\Gamma \left({\frac  {n-1-k}{2}}\right)}{\Gamma \left({\frac  {n-1}{2}}\right)}}{\mbox{exp}}\left(-{\frac  {\eta ^{2}}{2}}\right){\frac  {d^{k}}{d\eta ^{k}}}{\mbox{exp}}\left({\frac  {\eta ^{2}}{2}}\right),&{\mbox{if }}n-1 >k;\\{\mbox{Does not exist}},&{\mbox{if }}n-1 \leq k.\\\end{cases}}
\end{equation}

\subsection{Lower moments}\label{Lower_moments}
In particular, for $n > 3$, the mean and variance of the non-central t-distribution are defined and given by

\begin{equation}
\begin{aligned}
\mathbb{E} \left[T\right] &= \eta {\sqrt  {{\frac  {n-1 }{2}}}}{\frac  {\Gamma ( \frac{n-2}{2}) }{\Gamma ( \frac{n-1}{2} )}}  = \eta \,\, k_n \\
Var\left[T\right]	&= \frac  {(n-1) (1+\eta ^{2})}{n-3}-E\left[T\right] ^{2}
\end{aligned}
\end{equation}
\noindent where we have defined the constant $k_n = {\sqrt  {{\frac  {n-1 }{2}}}}{\frac  {\Gamma (
\frac{n-2}{2} )}{\Gamma (  \frac{n-1}{2})}}$. A good approximation for $k_n$ (related to the Wallis ratio $\frac  {1}{\sqrt \pi} \frac{\Gamma (
n+1 / 2)}{\Gamma ( n + 1 )} $) is $1 + \frac{3}{4n} +\frac{25}{32 n^2} + O (n^{-3})$ (see \cite{Digital_Library_Maths}). More on the Wallis ratio and multiple approximations can be for instance found in \cite{MORTICI2010}, \cite{Guo2013}, \cite{Qi_2015}, \cite{Lin_2017}). Another good approximation is $ \frac{1}{ 1 - \frac{ 3}{4n - 5}}$.

Instead of the constant $k_n$, it is quite common in statistical control litterature (see \cite{Duncan_1986} for instance) to use another constant called $c_4(n)$ (since there exists table for it) defined as follows 
\begin{eqnarray}
c_4(n) =  \sqrt  {\frac{2}{n-1}}  {\frac  {\Gamma( \frac{n}{2})}{\Gamma ( \frac{n-1}{2})}}
\end{eqnarray}

Using the traditional property of the Gamma function $\Gamma(n+1) = n \Gamma(n)$, it is immediate to see that $k_n = \frac{n-1}{n-2}  c_4(n)$.

\subsection{Asymptotic distribution}\label{Asymptotic_distribution}
When $n$ tends to infinity, the t-distribution denoted by $t$ tends to a normal distribution denoted by $N( \eta, \sigma) $ whose parameters are the first two moment of the non centered Student distribution. A better approximation is provided by \cite{Walk_2007} which states that
\begin{equation}\label{Walk_result}
\frac{t (1- \frac 1 {4(n-1)}) - \eta }{ \sqrt{1 + \frac{t^2}{2 (n-1) }} } \rightarrow N(0, 1)
\end{equation}

\section{Application to the Sharpe ratio}
\subsection{Non central distribution}
Let us apply these result to the eponymous Sharpe ratio denoted in the sequel (SR). Recall that it is defined as the ratio of expected excess return ($\bar R - R_f$) over the risk free rate $R_f$ to its standard deviation, $\sigma$:

\begin{equation}
SR = \frac{\bar R - R_f}{\sigma}
\end{equation}

The Sharpe ratio is simply the t-statistic divided by $\sqrt n$. In other terms, $\sqrt n SR $ follows a (centered or not) Student distribution \textbf{under the explicit assumption that the returns are normally distributed according to a normal distribution with mean $\boldsymbol{\mu_R = \bar R}$, variance $\boldsymbol{\sigma^2}$ and are i.i.d.}. Because $\bar R$ and $\sigma$ are unobservable, they are estimated using historical data as the population moments of the returns' distribution.
Hence, given historical returns $(R_1, R_2,..., R_n)$, the standard estimator for the Sharpe ratio (SR) is given by
\begin{equation}
\widehat{SR} = \frac{\hat{ \bar R} - R_f}{\hat \sigma}
\end{equation}

where 
\begin{equation}
\hat{ \bar R } =\frac{\sum_{i=1}^n R_i }{n}
\end{equation}

\begin{equation}
\hat \sigma  =  \sqrt{\frac { \sum_{i=1}^n \left( R_i - \hat{ \bar R} \right)^2 }{ n-1}} \quad  \text{or equivalently} \quad  \hat \sigma  =   \sqrt{\frac { \sum_{i=1}^n  \left(R_i ^2 - \hat{ \bar R}^2 \right) }{ n-1}}
\end{equation}
Precomputing $\hat{ \bar R }$ avoids a double summation in the definition of the volatility estimator. The $n-1$ divisor is for the estimator to be unbiased. 
Hence, the formula for SR in terms of the individual returns $\left( R_i \right)_{i=1..n}$ is given by:
\begin{equation}
\widehat{SR} = \frac{\sqrt{n-1} \sum_{i=1}^n \left( R_i  - R_f \right) }{  n \sqrt{\sum_{i=1}^n \left( R_i - \hat{ \bar R} \right)^2  } }
\end{equation}

An immediate application of the results of section \ref{Student} \, shows that $\sqrt n SR $ follows a non-central t-distribution with degree of freedom $n-1$. 
The non-centrality parameter $\mu$ is given by 
\begin{equation}
\eta  = \sqrt n  \; \frac{\mu_R - R_f}{\sigma} = \sqrt n  \; SR_{\infty}
\end{equation}

where the theoretical Sharpe ratio $SR_{\infty}$ is defined as $SR_{\infty} = \frac{\mu_R - R_f}{\sigma}$. 
Section \ref{Student} \, provides the exact distribution of the SR. This result that is quite basic is surprisingly almost absents in the financial literature although it was alluded in \cite{Miller_1978}. It is not for instance mentioned or noted in \cite{Lo_2002}. This result although simple is powerful yet as it provides various results for the Sharpe ratio. 
\begin{itemize}
\item First, it provides the real distribution (and leads obviously to the asymptotic distribution as the asymptotic distribution of the non-central t-distribution).
\item Second, all results for the t-statistic are directly transposable to SR, meaning we get for free any results about test, moments, cumulative distribution and cumulative density functions.
\item Third, as we are able to compute moments, we can immediately see that SR is biased estimator
\end{itemize}

\subsection{Moments}\label{Sharpe_moment}
We have in particular that the moments of the Sharpe ratio estimator are immediately provided by
\begin{eqnarray}\label{SR_M1}
\mathbb E \left[SR \right] 	&= &SR_{\infty} \; {\sqrt  {{\frac  {n-1 }{2 }}}}{\frac  {\Gamma ((n-2)/2)}{\Gamma (n-1 /2)}} \\
& = &  SR_{\infty} \; k_n
\end{eqnarray}

and

\begin{eqnarray}\label{SR_M2}
Var\left[SR \right]	&= &\frac{(n-1) (1+SR_{\infty} ^{2})}{n-3}-E\left[SR \right] ^{2}
\end{eqnarray}

Equation \ref{SR_M1} implies in particular that the empirical SR is biased with a bias term given by $k_n$, which is a result already noted in \cite{Miller_1978} and \cite{JobsonKorbie_1981}. As the constant $k_n$ is larger than one and multiplicative, the empirical Sharpe ratio will overestimate $SR_{\infty}$ when positive, and underestimate when negative. For one year of data and monthly data point, the bias $k_12$ is about 1.08, indicating an overestimation of 8 percents. We provide below in table \ref{tab:bias} the computation of bias for various value of $n$. The bias decreases rapidly as $n$ increases and is below 2 percents after roughly 3 years.

\begin{table}[H]
  \centering
  \caption{Sharpe ratio biais}
    \begin{tabular}{l | rrrrrrrr}
    \toprule
    n     & 3     & 6     & 12    & 24    & 36    & 48    & 60    & 120 \\
    \midrule
    \midrule
     bias  &    1.772  &    1.189  &      1.075  &      1.034  &      1.022  &      1.016  &      1.013  &        1.006   \\
    \bottomrule
    \end{tabular}%
  \label{tab:bias}%
\end{table}%

\subsection{A few results}
Using previous results \ref{Sharpe_moment}, we have the following result that extends the result of \cite{Lo_2002}

\begin{proposition}\label{prop:1}
SR under normal i.i.d. returns assumption is asymptotically normal in $n$ 

with standard deviation $\sigma_{IID,1}$ given by 
\begin{equation}\label{SIID_1}
\sigma_{IID,1}  = \sqrt{ 1+ \frac{ SR_{\infty} ^2} {2} }
\end{equation}
Another asymptotic approximation for the standard deviation (that is better for small n) is given by
\begin{equation}
\sigma_{IID,2}  = \sqrt{ \frac{1+ \frac{ SR_{\infty} ^2} {2} } {1- 1 / n}}
\end{equation}
Last but not least, a more precise estimation of the standard deviation  is given by
\begin{equation}
\sigma_{IID,3}  = \frac{ \sqrt{ 1 + \frac{ SR_{\infty} ^2} {2 (1-1/n)} }}{1 - \frac{1}{4 (n-1)}}
\end{equation}

The proposition means from a distribution point of view the following
\begin{equation}\label{eq:asymptotic_distribution}
\sqrt n ( \hat{SR} - SR_{\infty} ) \rightarrow N(0, \sigma_{IID,i})
\end{equation}

for $i=1$ to 3 depending on which version of the asymptotic standard deviation is used, where in the equation \eqref{eq:asymptotic_distribution} the normal distribution is parametrized in terms of the standard deviation.

\end{proposition}

\begin{proof}
Immediate using previous results and given in \ref{proof:1}
\end{proof}

Proposition \ref{prop:1} provides tighter bound for the asymptotic distribution than the one provided in \cite{Lo_2002}. In particular, it states that the $1- \alpha$ confidence interval for the empirical Sharpe ratio $\hat {SR}$ is given by
\begin{equation}
\hat {SR} \pm q_{\alpha / 2} \sigma_{IID,i}
\end{equation}
for $i=1$ to 3  where $q_{\alpha/2}$ is the $\alpha / 2$ quantile of the normal distribution in the asymptotic case. In the non asymptotic case, one needs to use the quantile of the non centered distribution. It is enlightening to compare the various estimator of the variance of SR. We provide in table \ref{tab:variance3} the computation of these volatilities for various values of $SR$ and $n$ according to formula $\sigma_{IID,3}$
For reference, we have also provided in appendix computation according to formula $\sigma_{IID,1}$ respectivly $\sigma_{IID,2}$ in table \ref{tab:variance1} and \ref{tab:variance2} as well as the difference with our best estimator of the variance of the Sharpe ratio.

\begin{table}[H]
  \centering
  \caption{Asymptotic variance for SR according to formula $\sigma_{IID,3}$}
    \begin{tabular}{|r|rrrrrrrr|}
    \toprule
    \multicolumn{1}{|l|}{SR} & 12    & 24    & 36    & 48    & 60    & 125   & 250   & 500 \\
    \midrule
    \midrule
    0.5   	&    0.308  &    0.217  &    0.177  &    0.153  &    0.137  &    0.095  &    0.067  &    0.047  \\
    0.75  	&    0.330  &    0.232  &    0.189  &    0.164  &    0.146  &    0.101  &    0.072  &    0.051  \\
    1     	&    0.359  &    0.252  &    0.205  &    0.177  &    0.159  &    0.110  &    0.078  &    0.055  \\
    1.25  	&    0.393  &    0.275  &    0.224  &    0.194  &    0.173  &    0.120  &    0.084  &    0.060  \\
    1.5   	&    0.431  &    0.301  &    0.245  &    0.212  &    0.189  &    0.131  &    0.092  &    0.065  \\
    1.75 	&    0.472  &    0.329  &    0.267  &    0.231  &    0.206  &    0.143  &    0.101  &    0.071  \\
    2     	&    0.515  &    0.359  &    0.291  &    0.252  &    0.225  &    0.155  &    0.110  &    0.078  \\
    2.25 	&    0.560  &    0.390  &    0.316  &    0.273  &    0.244  &    0.169  &    0.119  &    0.084  \\
    2.5   	&    0.606  &    0.421  &    0.342  &    0.296  &    0.264  &    0.182  &    0.129  &    0.091  \\
    2.75  	&    0.654  &    0.454  &    0.369  &    0.318  &    0.284  &    0.196  &    0.139  &    0.098  \\
    3     	&    0.702  &    0.487  &    0.395  &    0.341  &    0.305  &    0.210  &    0.149  &    0.105  \\
    \bottomrule
    \end{tabular}%
  \label{tab:variance3}%
\end{table}%

An interesting feature of SR variance is that it decreases asymptotically as $\frac{ 1 } { n }$ or equivalently as $\frac{ 1 } { n }$. SR variance is quite large as for instance for an empirical Sharpe ratio  of 1 and 12 month of data, the variance represents almost .36 or 36\% of it. If we take a quantile of 97.5\% whose student quantile is 2.201 (for 11 degree of freedom) as opposed to the normal well known quantile of 1.96, this implies that our Sharpe ratio lies between $1 \pm 2.201\times 0.359$ which provides as boundaries that the empirical Sharpe lies between  $[ 0.210, 1.790 ]$. This is very wide range and shows the imprecision of the Sharpe. Even for longer maturities like 5 years, the range is still wide: $1 \pm 2.00 \times 0.159 = [ 0.682, 1.318 ]$.

\subsection{Efficiency of the empirical Sharpe ratio}
Using the Frechet Darmois  Cramer Rao inequality, we can prove the following result that shows that the empirical SR is asymptotically efficient in the sense that it achieves the Cramer Rao bound
\begin{proposition}\label{prop:2}
under normal i.i.d. returns assumption, the estimator resulting from the empirical SR and the empirical variance is asymptotically efficient, meaning that it achieves the lower bound in terms of Cramer Rao bound given by
\begin{equation}
CRB = \frac{1}{n} { \left( 
\begin{array}{cc}
1+ SR_{\infty}^2 / 2 & -SR_{\infty} \sigma^2 \\
-SR_{\infty} \sigma^2  & 2 \sigma^4
\end{array}
\right) }
\end{equation}

\end{proposition}

\begin{proof}
Given in \ref{proof:2}
\end{proof}

\subsection{Weaker conditions}
\cite{Lo_2002}, \cite{Mertens_2002} and later \cite{Christie_2005} derived the asymptotic distribution of the Sharpe ratio under the more relaxed assumption of stationarity and ergodicity. \cite{Opdyke_2007} interestingly showed that the
derivation provided by \cite{Christie_2005} under the non-IID returns condition was in fact identical to the one provided by \cite{Mertens_2002}. 
\cite{Liu_2012} and \cite{Qi_2018} improved approximation accuracy to order $O(n^{-3/2})$ .

\subsection{Sample SR under AR(1) assumptions}
\subsubsection{Exact distribution for normal AR(1)}

The i.i.d. normal assumption for the return is far from being verified in practice. A more realistic set-up is to assume that the returns follow an AR(1) process defined as follows:
\begin{equation}\label{AR_assumptions}
\left\{ {
\begin{array}{l l l l }
R_t 			&  = & \mu + \epsilon_t 				& \quad t \geq 1 ; \\
\epsilon_t 	& = & \rho \epsilon_{t-1} + \sigma v_t 	& \quad t \geq 2 ; 
\end{array} } \right.
\end{equation}

where $v_t$ is an independent white noise processes (i.i.d. variables with zero mean and unit constant variance). To assume a stationary process, we impose
\begin{equation}
\lvert {\rho} \rvert  \leq 1
\end{equation}

It is easy to check that equation \ref{AR_assumptions} is equivalent to
\begin{equation}
\begin{array}{l l l l }
R_t 			&  = & \mu + \rho ( R_{t-1} - \mu )+ \sigma v_t 	& \quad t \geq 2 ; 
\end{array}
\end{equation}

We can also easily check that the variance and covariance of the returns are given by
\begin{equation}\label{moment2}
\begin{array}{l l l l }
V(R_t) & = & \frac {\sigma^2} {1-\rho^2} 								\; \;\;\;\; \text{for} \;  t \geq 1  \\
Cov(R_t, R_u ) & = & \frac {\sigma^2 \rho^{ \lvert {t -u} \rvert  }} {1-\rho^2} 	\; \;\;\;\; \text{for} \;  t,u \geq 1 
\end{array}
\end{equation}

Both expressions in \ref{moment2} are independent of time $t$ and the covariance only depends on $\lvert {t -u} \rvert$ implying that $R_t$ is a stationary process. If we now look at the empirical SR under these assumptions, it should converge to
\begin{equation}
\frac{ \mathbb{E}[{R_t}] - R_f } { \sqrt {var(R_t)}} = \frac{ \mu - R_f}{\sqrt { \frac{ \sigma^2 }{ 1- \rho^2}}} 
\end{equation}


\subsubsection{Impact of sub-sampling}
Another interesting feature that was first mentioned in \cite{Lo_2002}, is the impact of computing annual Sharpe using monthly data. This can be formalized as follows:
Let us define the $q$ period return as
\begin{equation}
R_t(q) \equiv R_t + R_{t-1} + \ldots+ R_{t-q+1} 
\end{equation}
where in our definition, we have ignored the effects of compounding for computational efficiency\footnote{Of course, the exact expression for compounding returns is
$ R_t(q) \equiv \prod_{i=0}^{q-1}( 1 + R_{t-j} )-1$. But the sake of clarity, we can ignore the compounding effect in our section as this will be second order effect. Equally, we could use log or continuously compounded returns defined as $\log(P_t/P_{t-1})$ in which case, our definition would be exact }. We are interested in measuring effect of auto correlation and heteroscedasticity of returns on the Sharpe ratio and the impact of using for instance monthly return for computing annual Sharpe. 
Let us denote by $SR(q)$ the SR computed with $q$ period returns. Its limit is defined as follows:
\begin{eqnarray}
SR(q) & = & \frac{ \mathbb{E}[ R_t(q)  ] - R_{f}}{ \sqrt{ Var[R_t(q)] }   }
\end{eqnarray}

Let us denote the returns mean by $\mu$, the auto correlation by $\rho_{u,v} = Corr(R_{u}, R_{v})$  and the returns variance by $\sigma_{\infty}^2 = \lim_{t \rightarrow \infty} Var[R_t]$. It is interesting to see the linkage between the $q$ period  SR denoted by $SR(q)$ and the regular SR denoted by $ SR  =  \frac{ \mu - R_{f}}{ \sigma_{\infty} } $. This is the subject of the following proposition

\begin{proposition}\label{prop:3}
The ratio between the $q$ period returns $SR(q)$ and the regular SR is the following:
\begin{equation}\label{SRq_eq1}
\frac{SR(q)}{SR} = \frac{q \sigma_{\infty} }{ \sqrt{ \sum_{i=0}^{q-1} \sigma^2_{t-i} + 2  \sum_{k=1}^{q-1} \sum_{i = 0 }^{q-1-k } \rho_{t-i, t-i-k} \sigma_{t-i} \sigma_{t-i-k} } }
\end{equation}
If the return process is stationary with a constant variance  $\sigma^2 =  Var[R_t] =\sigma_{\infty}^2 $ and stationary correlation denoted by $\rho_{v-u} = Corr( R_{u},R_{v})$, this relationship simplifies to
\begin{equation}\label{SRq_eq2}
\frac{SR(q)}{SR}= \sqrt{ \frac{ q} { 1 + 2  \sum_{k=1}^{q-1} (q-k) \rho_{k} } }
\end{equation}
If in addition, the returns follow an AR(1) process $\rho_k=\rho^k$, equation \ref{SRq_eq2} becomes
\begin{equation}\label{SRq_eq3}
\frac{SR(q)}{SR} = \sqrt{ \frac{ q} { 1 + \frac{ 2 \rho }{ 1-\rho} \left( 1 - \frac{ 1 - \rho^q}{q (1-\rho)} \right) } }
\end{equation}

If the returns are non correlated ($\rho_k=0$), equation \ref{SRq_eq3} becomes
\begin{equation}\label{SRq_eq4}
\frac{SR(q)}{SR} = \sqrt{ q}
\end{equation}
The last equation is the so called square root rule that states that the annual Sharpe is equal to $\sqrt{12}$ the monthly Sharpe.
\end{proposition}

\begin{proof}
Given in \ref{proof:3}
\end{proof}

A numerical result is provided in table \ref{tab:compoudingSharpe}. Proposition \ref{prop:3} is important as it shows that the compounding effect (computing annual Sharpe with monthly return) can have some large impact when using the square root rule to convert the monthly Sharpe to the annual Sharpe in presence of autocorrelation and heteroscedasticity.


\section{Conclusion}
Even if Sharpe ratio is the norm of funds financial analysis, we have shown that its empirical estimate has various bias that makes its usage for ranking questionable. Generally, there is well reported literature that Sharpe ratios are skewed
to the left, fat tailed and very sensitive to small samples (see \cite{Goetzmann_2012}). Not surprisingly, the bias are highly influenced by the statistical properties of the returns time series (and in particular auto-correlation and heteroscedasticy). Our work extends previous results in terms of the real distribution of the empirical Sharpe ratio and give as a by-product standards results obtained in the Sharpe ratio financial literature. This work advocates for more substantial analysis when ever comparing funds and in particular a good understanding of investment style to identify potential skew and autocorrelation in fund presented performance monthly returns. This also encourages to use various other performance ratios to analyze deeply funds performance.

\clearpage

\appendix
\section{Various Proofs}

\subsubsection{Proof of Proposition \ref{prop:1}}\label{proof:1}
Section \ref{Asymptotic_distribution} states that $\sqrt n \; SR$ tends asymptotically to a normal distribution whose first two moments are given by the asymptotic limit of \ref{SR_M1} and \ref{SR_M2}.
\ref{SR_M2} provided the exact formula for the standard deviation of $\sqrt n SR$. Denoting by $\sigma_{IID}^2$ the variance of the random variable $\sqrt n SR$, and using the expression for the variance of $SR$ thanks to equation \eqref{SR_M2}, we get using a Taylor expansion in power of $\frac 1 n $:
\begin{eqnarray} 
\frac{\sigma_{IID}^2}{n} & =  &\frac{(n-1) (1+SR_{\infty} ^{2})}{n-3} -  (SR_{\infty} \; k_n)^2 \\
& = & \frac 1 n + SR_{\infty} ^{2} (({1 - \frac 1 n })({1  + \frac 3 n }) - (1+ \frac {6} {4n}) ) + O( \frac{1}{n^2}) \\
& = & \frac 1 n + \frac{ SR_{\infty} ^{2} } { 2 n } + O( \frac{1}{n^2})
\end{eqnarray} 

This is the result obtained by \cite{Lo_2002}. It is obviously equivalent asymptotically to $\sigma_{IID,2}$. 
If we use the result provided by \ref{Walk_result}, we have another approximation for the standard deviation given by
\begin{eqnarray} 
\sigma_{IID,3}  = \frac{ \sqrt{ \frac{1}{n} + \frac{ SR_{\infty} ^2} {2 (n-1)} }}{1 - \frac{1}{4 (n-1)}}
\end{eqnarray} 
 
\subsubsection{Proof of Proposition \ref{prop:2}}\label{proof:2}
The log-likelihood of i.i.d. returns with normal distribution with unknown Sharpe ratio $s$ and variance $v$ is given by
\begin{equation}
\mathcal{L}(s,v) = -\frac n 2  \log{ (2 \pi v )} - \sum_{i=1}^{n}\frac{ (R_i - R_f -  s v^{1/2} ) ^2 }{2 v }
\end{equation}

The Fisher information for the estimator resulting from the empirical SR and the empirical variance is computed as the expected opposite of the second order derivative of the log likelihood with respect to the parameters denoted by $\theta = ( s, v ) ^T $
\begin{equation}
\mathcal{I}(s,v) = -\mathbb{E}\left( \frac{ \partial^2 \mathcal{L}(s,v) }{ \partial^2 \theta } |\theta \right)
\end{equation}
This log-likelihood is computed for the parameters that maximizes the log likelihood and given by 
\begin{eqnarray}\label{mle}
\bar s & = &  \frac { \sum_{i=1}^n (R_i - R_f ) }{ n \; v^{1/2} } \\
\bar v & = &  \frac {\sum_{i=1}^n (R_i - R_f -  \bar{s} v^{1/2} )^2  }{n } 
\end{eqnarray}

Straight computation using \ref{mle} leads to 
\begin{eqnarray}
\frac{ \partial^2 \mathcal{L}(s,v) }{ \partial^2 s} 			& = &  -n \\
\frac{ \partial^2 \mathcal{L}(s,v) }{ \partial v \partial s } 		& = &  -\frac{s}{2 v }\\
\frac{ \partial^2 \mathcal{L}(s,v) }{ \partial^2 v} 			& = &  -\frac{ 2+ s^2 }{4 v^2}
\end{eqnarray}

This implies that the Fisher information is given by
\begin{equation}
\mathcal{I}(s,v) = n  \left( 
\begin{array}{cc}
1 				& \frac{s}{2 v } \\
\frac{s}{2 v }		& \frac{ 2+ s^2 }{4 v^2}
\end{array}
\right) 
\end{equation}

Inverting is trivial and leads to the Cramer Rao bound as follows:
\begin{equation}
CRB = \mathcal{I}^{-1} (s,v)  = \frac{1}{n} { \left( 
\begin{array}{cc}
1+ s^2 / 2 & -s \sigma^2 \\
-s \sigma^2  & 2 \sigma^4
\end{array}
\right) }
\end{equation}

Now if we consider the estimator given by the empirical Sharpe ratio and the variance $\left[ \hat{SR}, \hat{v} \right]^T$. It is an unbiased estimator of  $\left[ s, v \right]^T$. 
Equations \ref{SR_M2} and \ref{SR_M1} state that the variance of $ \hat{SR}$ is given by 
\begin{eqnarray}
Var\left[SR \right]	&= &\frac{(n-1) (1+SR_{\infty} ^{2})}{n-3}- ( SR_{\infty} \; k_n) ^{2}
\end{eqnarray}

whose asymptotic limit is (see \ref{SIID_1}) $\frac{1+ s^2/2}{n}$.  The variance of the empirical variance is well known and given by $\frac{ 2 \sigma^4}{n-1}$ that is asymptotically equivalent to $\frac{ 2 \sigma^4}{n}$. The covariance term between the empirical SR and the variance is more involved and can be found in \cite{Mertens_2002} or \cite{Pav_2016} and is asymptomatically equivalent to  by $\frac{-s \sigma^2 }{n}$ which concludes the proof as the  estimator given by the empirical Sharpe ratio and the variance achieves the Cramer-Rao lower
bound asymptotically. 
\qed

\subsubsection{Proof of Proposition \ref{prop:3}}\label{proof:3}
The numerator of the empirical SR converges to $q (\mu -  R_{f})$ where $\mu$ denotes the returns mean. The denominator can be expanded to measure the impact of (auto)correlation between returns as follows:

\begin{eqnarray}
Var[R_t(q)]  &= & \sum_{i=0}^{q-1} \sum_{j=0}^{q-1} Cov( R_{t-i}, R_{t-j} ) \\
&= & \sum_{i=0}^{q-1} Var(R_{t-i} ) + 2  \sum_{i=0}^{q-1} \sum_{j=i+1}^{q-1} Cov( R_{t-i}, R_{t-j} )
\end{eqnarray}

Using the change of variable, $k=j-i$, denoting by $\rho_{u,v} = Corr(R_{u}, R_{v})$ and by $\sigma^2_{u}=Var(Var(R_{u})$, we get the following expression:
\begin{equation} \label{GeneralEqVar}
Var[R_t(q)]  =  \sum_{i=0}^{q-1} \sigma^2_{t-i} + 2  \sum_{k=1}^{q-1} \sum_{i = 0 }^{q-1-k } \rho_{t-i, t-i-k} \sigma_{t-i} \sigma_{t-i-k}
\end{equation}

Computed the fraction $\frac{SR(q)}{SR}$ leads to the result of equation \ref{SRq_eq1}. This is the most general formula that extends the one of  \cite{Lo_2002}. 

If the return process is stationary with a constant variance, then
\begin{itemize}
\item $\sigma_{u}$ does not depend on $u$ and can be denoted by $\sigma$, 
\item $\rho_{u,v}$ only depends on the absolute difference between $u$ and $v$ and is written for $v \geq u$ as $\rho_{v-u}$
\end{itemize}

Then equation \ref{GeneralEqVar} becomes
\begin{eqnarray}\label{GeneralEqVar2}
Var[R_t(q)]  &= &  \sigma^2  \left( q + 2  \sum_{k=1}^{q-1} (q-k) \rho_{k} \right)
\end{eqnarray}

Again, computing the fraction $\frac{SR(q)}{SR}$ leads to the result of equation \ref{SRq_eq2}. 

For AR(1) process, we have $\rho_k=\rho^k$ and  equation \ref{GeneralEqVar2} becomes
\begin{eqnarray}
Var[R_t(q)]  &= &  \sigma^2  \left[ q + \frac{ 2 \rho }{ 1-\rho} \left( 1 - \frac{ 1 - \rho^q}{q (1-\rho)} \right) \right]
\end{eqnarray}

where we have used $\sum_{k=0}^{q-2}  \rho^k = \frac{1- \rho^{q-1} }{1-\rho}$ 
and $\sum_{k=0}^{q-1} k \rho^{k-1} = \frac{ \frac{1- \rho^{q}}{1-\rho} - q \rho^{q-1}}{ (1-\rho) }$. 
Equation \ref{SRq_eq4} is immediate.
\qed

\section{Numerical applications}
\subsection{Variance computation}
We provide here values of the variance of the Sharpe ratio according to formula $\sigma_{IID,1}$ respectivly $\sigma_{IID,2}$ in table \ref{tab:variance2} and \ref{tab:variance3} as well as the difference with our best estimator of the variance of the Sharpe ratio.

\begin{table}[H]
  \centering
  \caption{Asymptotic variance for SR according to formula $\sigma_{IID,1}$}
    \begin{tabular}{|r|rrrrrrrr|}
    \toprule
    \multicolumn{1}{|l|}{SR} & 12    & 24    & 36    & 48    & 60    & 125   & 250   & 500 \\
    \midrule
    \midrule
    0.5   	&    0.306  &    0.217  &    0.177  &    0.153  &    0.137  &    0.095  &    0.067  &    0.047  \\
    0.75 	&    0.327  &    0.231  &    0.189  &    0.163  &    0.146  &    0.101  &    0.072  &    0.051  \\
    1     	&    0.354  &    0.250  &    0.204  &    0.177  &    0.158  &    0.110  &    0.077  &    0.055  \\
    1.25  	&    0.385  &    0.272  &    0.222  &    0.193  &    0.172  &    0.119  &    0.084  &    0.060  \\
    1.5   	&    0.421  &    0.298  &    0.243  &    0.210  &    0.188  &    0.130  &    0.092  &    0.065  \\
    1.75  	&    0.459  &    0.325  &    0.265  &    0.230  &    0.205  &    0.142  &    0.101  &    0.071  \\
    2     	&    0.500  &    0.354  &    0.289  &    0.250  &    0.224  &    0.155  &    0.110  &    0.077  \\
    2.25  	&    0.542  &    0.384  &    0.313  &    0.271  &    0.243  &    0.168  &    0.119  &    0.084  \\
    2.5   	&    0.586  &    0.415  &    0.339  &    0.293  &    0.262  &    0.182  &    0.128  &    0.091  \\
    2.75 	&    0.631  &    0.446  &    0.364  &    0.316  &    0.282  &    0.196  &    0.138  &    0.098  \\
    3     	&    0.677  &    0.479  &    0.391  &    0.339  &    0.303  &    0.210  &    0.148  &    0.105  \\
    \bottomrule
    \end{tabular}%
  \label{tab:variance1}%
\end{table}%

\begin{table}[H]
  \centering
  \caption{Asymptotic variance for SR according to formula $\sigma_{IID,2}$}
    \begin{tabular}{|r|rrrrrrrr|}
    \toprule
    \multicolumn{1}{|l|}{SR} & 12    & 24    & 36    & 48    & 60    & 125   & 250   & 500 \\
    \midrule
    \midrule
    0.5   	&    0.320  &    0.221  &    0.179  &    0.155  &    0.138  &    0.095  &    0.067  &    0.047  \\
    0.75  	&    0.341  &    0.236  &    0.191  &    0.165  &    0.147  &    0.102  &    0.072  &    0.051  \\
    1     	&    0.369  &    0.255  &    0.207  &    0.179  &    0.159  &    0.110  &    0.078  &    0.055  \\
    1.25 	&    0.402  &    0.278  &    0.226  &    0.195  &    0.174  &    0.120  &    0.085  &    0.060  \\
    1.5   	&    0.440  &    0.304  &    0.246  &    0.213  &    0.190  &    0.131  &    0.092  &    0.065  \\
    1.75  	&    0.480  &    0.332  &    0.269  &    0.232  &    0.207  &    0.143  &    0.101  &    0.071  \\
    2     	&    0.522  &    0.361  &    0.293  &    0.253  &    0.225  &    0.156  &    0.110  &    0.078  \\
    2.25 	&    0.567  &    0.392  &    0.318  &    0.274  &    0.245  &    0.169  &    0.119  &    0.084  \\
    2.5   	&    0.612  &    0.423  &    0.343  &    0.296  &    0.264  &    0.182  &    0.129  &    0.091  \\
    2.75 	&    0.659  &    0.456  &    0.370  &    0.319  &    0.285  &    0.196  &    0.139  &    0.098  \\
    3     	&    0.707  &    0.489  &    0.396  &    0.342  &    0.305  &    0.211  &    0.149  &    0.105  \\
    \bottomrule
    \end{tabular}%
  \label{tab:variance2}%
\end{table}%

\begin{table}[H]
  \centering
  \caption{Difference between asymptotic variance for SR according to formula $\sigma_{IID,1}$ and $\sigma_{IID,3}$}
    \begin{tabular}{|r|rrrrrrrr|}
    \toprule
    \multicolumn{1}{|l|}{SR} & 12    & 24    & 36    & 48    & 60    & 125   & 250   & 500 \\
    \midrule
    \midrule
    0.5   	& 1.21\% & 0.41\% & 0.22\% & 0.14\% & 0.10\% & 0.03\% & 0.01\% & 0.00\% \\
    0.75 	& 1.13\% & 0.39\% & 0.21\% & 0.13\% & 0.10\% & 0.03\% & 0.01\% & 0.00\% \\
    1     	& 1.04\% & 0.36\% & 0.19\% & 0.12\% & 0.09\% & 0.03\% & 0.01\% & 0.00\% \\
    1.25  	& 0.95\% & 0.33\% & 0.18\% & 0.11\% & 0.08\% & 0.03\% & 0.01\% & 0.00\% \\
    1.5   	& 0.87\% & 0.30\% & 0.16\% & 0.10\% & 0.07\% & 0.02\% & 0.01\% & 0.00\% \\
    1.75 	& 0.80\% & 0.27\% & 0.15\% & 0.10\% & 0.07\% & 0.02\% & 0.01\% & 0.00\% \\
    2     	& 0.73\% & 0.25\% & 0.14\% & 0.09\% & 0.06\% & 0.02\% & 0.01\% & 0.00\% \\
    2.25 	& 0.67\% & 0.23\% & 0.13\% & 0.08\% & 0.06\% & 0.02\% & 0.01\% & 0.00\% \\
    2.5   	& 0.62\% & 0.21\% & 0.12\% & 0.07\% & 0.05\% & 0.02\% & 0.01\% & 0.00\% \\
    2.75 	& 0.58\% & 0.20\% & 0.11\% & 0.07\% & 0.05\% & 0.02\% & 0.01\% & 0.00\% \\
    3     	& 0.54\% & 0.19\% & 0.10\% & 0.06\% & 0.05\% & 0.02\% & 0.01\% & 0.00\% \\
    \bottomrule
    \end{tabular}%
  \label{tab:diff_variance1}%
\end{table}%

\begin{table}[H]
  \centering
  \caption{Difference between asymptotic variance for SR according to formula $\sigma_{IID,2}$ and $\sigma_{IID,3}$}
    \begin{tabular}{|r|rrrrrrrr|}
    \toprule
    \multicolumn{1}{|l|}{SR} & 12    & 24    & 36    & 48    & 60    & 125   & 250   & 500 \\
    \midrule
    \midrule
    0.5   	& 1.21\% & 0.41\% & 0.22\% & 0.14\% & 0.10\% & 0.03\% & 0.01\% & 0.00\% \\
    0.75  	& 1.13\% & 0.39\% & 0.21\% & 0.13\% & 0.10\% & 0.03\% & 0.01\% & 0.00\% \\
    1     	& 1.04\% & 0.36\% & 0.19\% & 0.12\% & 0.09\% & 0.03\% & 0.01\% & 0.00\% \\
    1.25  	& 0.95\% & 0.33\% & 0.18\% & 0.11\% & 0.08\% & 0.03\% & 0.01\% & 0.00\% \\
    1.5   	& 0.87\% & 0.30\% & 0.16\% & 0.10\% & 0.07\% & 0.02\% & 0.01\% & 0.00\% \\
    1.75  	& 0.80\% & 0.27\% & 0.15\% & 0.10\% & 0.07\% & 0.02\% & 0.01\% & 0.00\% \\
    2     	& 0.73\% & 0.25\% & 0.14\% & 0.09\% & 0.06\% & 0.02\% & 0.01\% & 0.00\% \\
    2.25  	& 0.67\% & 0.23\% & 0.13\% & 0.08\% & 0.06\% & 0.02\% & 0.01\% & 0.00\% \\
    2.5   	& 0.62\% & 0.21\% & 0.12\% & 0.07\% & 0.05\% & 0.02\% & 0.01\% & 0.00\% \\
    2.75  	& 0.58\% & 0.20\% & 0.11\% & 0.07\% & 0.05\% & 0.02\% & 0.01\% & 0.00\% \\
    3     	& 0.54\% & 0.19\% & 0.10\% & 0.06\% & 0.05\% & 0.02\% & 0.01\% & 0.00\% \\
    \bottomrule
    \end{tabular}%
  \label{tab:diff_variance2}%
\end{table}%

\subsection{Compounding effect}

\begin{table}[H]
  \centering
  \caption{Compounding effect for AR(1) process for the Sharpe}
    \begin{tabular}{|r|rrrrrrrrrr|}
    \toprule
    $\rho$ \textbackslash  \, $q $   & 2     & 3     & 4     & 6     & 12    & 24    & 36    & 48    & 125   & 250 \\
    \midrule
    \midrule
    90\%  &         1.026  &         1.046  &         1.065  &         1.102  &         1.207  &         1.408  &         1.597	&         1.773  &         2.668  &         3.698  \\
    80\%  &         1.054  &         1.097  &         1.137  &         1.213  &         1.427  &         1.808  &         2.136	&         2.424  &         3.795  &         5.318  \\
    70\%  &         1.085  &         1.152  &         1.215  &         1.333  &         1.654  &         2.187  &         2.622	&         2.997  &         4.749  &         6.679  \\
    60\%  &         1.118  &         1.213  &         1.300  &         1.462  &         1.885  &         2.551  &         3.081	&         3.534  &         5.633  &         7.936  \\
    50\%  &         1.155  &         1.279  &         1.393  &         1.600  &         2.121  &         2.910  &         3.530	&         4.057  &         6.490  &         9.153  \\
    40\%  &         1.195  &         1.353  &         1.494  &         1.748  &         2.364  &         3.273  &         3.981 	&         4.581  &         7.347  &      10.371  \\
    30\%  &         1.240  &         1.433  &         1.605  &         1.905  &         2.615  &         3.645  &         4.444  	&         5.119  &         8.226  &      11.618  \\
    20\%  &         1.291  &         1.523  &         1.725  &         2.073  &         2.879  &         4.035  &         4.928  	&         5.682  &         9.144  &      12.921  \\
    10\%  &         1.348  &         1.622  &         1.857  &         2.254  &         3.160  &         4.450  &         5.442  	&         6.280  &      10.121  &      14.308  \\
    0\%   &         1.414  &         1.732  &         2.000  &         2.449  &         3.464  &         4.899   &         6.000  	&         6.928  &      11.180  &      15.811  \\
    -10\% &         1.491  &         1.853  &         2.157  &         2.664  &         3.798  &         5.393  &         6.615  	&         7.643  &      12.350  &      17.473  \\
    -20\% &         1.581  &         1.987  &         2.331  &         2.901  &         4.171  &         5.949  &         7.306  	&         8.449  &      13.670  &      19.349  \\
    -30\% &         1.690  &         2.132  &         2.527  &         3.169  &         4.596  &         6.586  &         8.103  	&         9.377  &      15.196  &      21.519  \\
    -40\% &         1.826  &         2.287  &         2.752  &         3.477  &         5.093  &         7.339  &         9.046  	&      10.480  &      17.014  &      24.106  \\
    -50\% &         2.000  &         2.449  &         3.024  &         3.843  &         5.692  &         8.259  &      10.205  	&      11.837  &      19.262  &      27.313  \\
    -60\% &         2.236  &         2.611  &         3.371  &         4.300  &         6.444  &         9.436  &      11.699  	&      13.593  &      22.195  &      31.505  \\
    -70\% &         2.582  &         2.762  &         3.860  &         4.922  &         7.449  &       11.047  &      13.768  	&      16.040  &      26.327  &      37.434  \\
    -80\% &         3.162  &         2.887  &         4.663  &         5.909  &         8.961  &       13.505  &      16.983  	&      19.884  &      32.960  &      47.018  \\
    -90\% &         4.472  &         2.970  &         6.472  &         8.095  &      12.064  &         18.289  &      23.325  	&      27.613  &      46.986  &      67.650  \\
    \bottomrule
    \end{tabular}
  \label{tab:compoudingSharpe}
\end{table}%

Compared to the standard square root rule, we can compute the ratio of the real factor and $\sqrt d $ given below

\begin{table}[H]
  \centering
  \caption{This table provides the ratio of the correct compounding effect and the usual square root constant. More precisely, we compute $ { 1 + \frac{ 2 \rho }{ 1-\rho} \left( 1 - \frac{ 1 - \rho^q}{q (1-\rho)} \right) }$ }
    \begin{tabular}{|r|rrrrrrrrrr|}
    \toprule
    $\rho$  \textbackslash \, $q $   & 2     & 3     & 4     & 6     & 12    & 24    & 36    & 48    & 125   & 250 \\
    \midrule
    \midrule
    90\%  &         1.378  &         1.655  &         1.877  &         2.223  &         2.870  &         3.478  &         3.757  &         3.908  &         4.190  &         4.276  \\
    80\%  &         1.342  &         1.579  &         1.760  &         2.020  &         2.428  &         2.709  &         2.809  &         2.858  &         2.946  &         2.973  \\
    70\%  &         1.304  &         1.503  &         1.647  &         1.838  &         2.095  &         2.240  &         2.288  &         2.311  &         2.354  &         2.367  \\
    60\%  &         1.265  &         1.428  &         1.539  &         1.676  &         1.837  &         1.920  &         1.947  &         1.961  &         1.985  &         1.992  \\
    50\%  &         1.225  &         1.354  &         1.436  &         1.531  &         1.633  &         1.683  &         1.700  &         1.708  &         1.723  &         1.727  \\
    40\%  &         1.183  &         1.281  &         1.339  &         1.402  &         1.466  &         1.497  &         1.507  &         1.512  &         1.522  &         1.525  \\
    30\%  &         1.140  &         1.208  &         1.246  &         1.286  &         1.325  &         1.344  &         1.350  &         1.353  &         1.359  &         1.361  \\
    20\%  &         1.095  &         1.137  &         1.159  &         1.181  &         1.203  &         1.214  &         1.218  &         1.219  &         1.223  &         1.224  \\
    10\%  &         1.049  &         1.068  &         1.077  &         1.087  &         1.096  &         1.101  &         1.102  &         1.103  &         1.105  &         1.105  \\
    0\%   &         1.000  &         1.000  &         1.000  &         1.000  &         1.000  &         1.000  &         1.000  &         1.000  &         1.000  &         1.000  \\
    -10\% &         0.949  &         0.935  &         0.927  &         0.920  &         0.912  &         0.908  &         0.907  &         0.906  &         0.905  &         0.905  \\
    -20\% &         0.894  &         0.872  &         0.858  &         0.844  &         0.831  &         0.824  &         0.821  &         0.820  &         0.818  &         0.817  \\
    -30\% &         0.837  &         0.812  &         0.792  &         0.773  &         0.754  &         0.744  &         0.740  &         0.739  &         0.736  &         0.735  \\
    -40\% &         0.775  &         0.757  &         0.727  &         0.704  &         0.680  &         0.668  &         0.663  &         0.661  &         0.657  &         0.656  \\
    -50\% &         0.707  &         0.707  &         0.661  &         0.637  &         0.609  &         0.593  &         0.588  &         0.585  &         0.580  &         0.579  \\
    -60\% &         0.632  &         0.663  &         0.593  &         0.570  &         0.538  &         0.519  &         0.513  &         0.510  &         0.504  &         0.502  \\
    -70\% &         0.548  &         0.627  &         0.518  &         0.498  &         0.465  &         0.443  &         0.436  &         0.432  &         0.425  &         0.422  \\
    -80\% &         0.447  &         0.600  &         0.429  &         0.415  &         0.387  &         0.363  &         0.353  &         0.348  &         0.339  &         0.336  \\
    -90\% &         0.316  &         0.583  &         0.309  &         0.303  &         0.287  &         0.268  &         0.257  &         0.251  &         0.238  &         0.234  \\
    \bottomrule
    \end{tabular}%
  \label{tab:addlabel}%
\end{table}%

\clearpage


\bibliographystyle{jfe}
\bibliography{mybib}

\begin{thebibliography}{29}
\expandafter\ifx\csname natexlab\endcsname\relax\def\natexlab#1{#1}\fi

\bibitem[{Bondesson(1974)}]{Bondesson_1974}
Bondesson, L., 1974. Characterizations of probability laws through constant
  regression. Z.Wahrscheinlichkeitstheorie verw. Gebiete pp. 93--115.

\bibitem[{Bondesson(1983)}]{Bondesson_1983}
Bondesson, L., 1983. When is the t-statistic t-distributed. Sankhya, Ser. A pp.
  338--345.

\bibitem[{Christie(2005)}]{Christie_2005}
Christie, S., 2005. Is the sharpe ratio useful in asset allocation? MAFC
  Research Papers .

\bibitem[{Duncan(1986)}]{Duncan_1986}
Duncan, A.~J., 1986. Quality control and industrial statistics.

\bibitem[{Efron(1969)}]{Efron_1969}
Efron, B., 1969. Student's t-test under symmetry conditions. J. Amer. Statist.
  Assoc. pp. 1278--1302.

\bibitem[{Fang et~al.(2001)Fang, Yang, and Kotz}]{Fang_2001}
Fang, K., Yang, Z., Kotz, S., 2001. Generation of multivariate distributions by
  vertical density representation. Statistics pp. 281--293.

\bibitem[{Gallagher(2011)}]{Digital_Library_Maths}
Gallagher, P.~D., 2011. Digital library of mathematical functions.

\bibitem[{Goetzmann et~al.(2002)Goetzmann, Ingersoll, Spiegel, and
  Welch}]{Goetzmann_2012}
Goetzmann, W., Ingersoll, J., Spiegel, M.~I., Welch, I., 2002. Sharpening
  sharpe ratios. NBER Working Paper No. 9116 .

\bibitem[{Guo et~al.(2013)Guo, Xu, and Qi}]{Guo2013}
Guo, S., Xu, J.-G., Qi, F., 2013. Some exact constants for the approximation of
  the quantity in the wallis' formula. Journal of Inequalities and Applications
  2013, 67.

\bibitem[{Helmert(1876)}]{Helmert_1876}
Helmert, F.~R., 1876. Uber die wahrscheinlichkeit der potenzsummen der
  beobachtungsfehler und uber einige damit in zusammenhang stehende fragen.
  Math. Phys pp. 192--218.

\bibitem[{Jobson and Korkie(1981)}]{JobsonKorbie_1981}
Jobson, J.~D., Korkie, B.~M., 1981. Performance hypothesis testing with the
  sharpe and treynor measures. The Journal of Finance 36, 889--908.

\bibitem[{Kagan et~al.(1973)Kagan, Linnik, and Rao}]{Kagan_1973}
Kagan, A., Linnik, Y., Rao, C., 1973. Characterization problems in mathematical
  statistics.

\bibitem[{Lenth(1989)}]{Lenth_1989}
Lenth, R.~V., 1989. Algorithm as 243: Cumulative distribution function of the
  non-central t distribution. Journal of the Royal Statistical Society, Series
  C pp. 185--189.

\bibitem[{Lin et~al.(2017)Lin, Ma, and Chen}]{Lin_2017}
Lin, L., Ma, W.-C., Chen, C.-P., 2017. Pade approximant related to the wallis
  formula. Journal of Inequalities and Applications .

\bibitem[{Liu et~al.(2012)Liu, Rekkas, and Wong}]{Liu_2012}
Liu, Y., Rekkas, M., Wong, A., 2012. Inference for the sharpe ratio using a
  likelihood-based approach. Journal of Probability and Statistics .

\bibitem[{Lo(2002)}]{Lo_2002}
Lo, A.~W., 2002. The statistics of sharpe ratios. Financial Analysts Journal
  58, No. 4, 36--52.

\bibitem[{Luroth(1876)}]{Luroth_1876}
Luroth, J., 1876. Vergleichung von zwei werten des wahrscheinlichen fehlers.
  Astron Nachr 14, 209--220.

\bibitem[{Mauldon(1956)}]{Mauldon_1956}
Mauldon, J., 1956. Characterizing properties of statistical distributions.
  Quart. J. Math. pp. 155--160.

\bibitem[{Mertens(2002)}]{Mertens_2002}
Mertens, E., 2002. Comments on variance of the iid estimator in lo. Technical
  report, Working Paper University of Basel .

\bibitem[{Miller and Gehr(1978)}]{Miller_1978}
Miller, R.~E., Gehr, A.~K., 1978. Sample size bias and sharpe's performance
  measure: A note. Journal of Financial and Quantitative Analysis pp. 943--946.

\bibitem[{Mortici(2010)}]{MORTICI2010}
Mortici, C., 2010. New approximation formulas for evaluating the ratio of gamma
  functions. Mathematical and Computer Modelling 52, 425 -- 433.

\bibitem[{Opdyke(2007)}]{Opdyke_2007}
Opdyke, J.~D., 2007. Comparing sharpe ratios: so where are the p-values?
  Journal of Asset Management 8, 308–336.

\bibitem[{Pav(2016)}]{Pav_2016}
Pav, S.~E., 2016. Notes on the sharpe ratio. {/cran.r-project.org} .

\bibitem[{Pearson(1895)}]{Pearson_1985}
Pearson, K., 1895. Contributions to the mathematical theory of evolution, ii:
  Skew variation in homogeneous material. Philosophical Transactions of the
  Royal Society pp. 343--414.

\bibitem[{Qi and Mortici(2015)}]{Qi_2015}
Qi, F., Mortici, C., 2015. Some best approximation formulas and inequalities
  for the wallis ratio. Applied Mathematics and Computation 253, 363 -- 368.

\bibitem[{Qi et~al.(2018)Qi, Rekkas, and Wong}]{Qi_2018}
Qi, J., Rekkas, M., Wong, A., 2018. Highly accurate inference on the sharpe
  ratio for autocorrelated return data. Journal of Statistical and Econometric
  Methods .

\bibitem[{Sharpe(1966)}]{Sharpe_1966}
Sharpe, W.~F., 1966. Mutual fund performance. Journal of Business pp. 119--138.

\bibitem[{{Student alias W. S. Gosset}(1908)}]{Student_1908}
{Student alias W. S. Gosset}, 1908. The probable error of a mean. Biometrika
  pp. 1--25.

\bibitem[{Walck(2007)}]{Walk_2007}
Walck, C., 2007. Hand-book on Statistical distribution for experimentalists.
  Stockholm University.

\end{thebibliography}

\clearpage

\end{document}